\documentclass[11pt]{article}
\usepackage[utf8]{inputenc}
\usepackage[T2A]{fontenc}
\usepackage[english]{babel}
\usepackage{amsmath,amssymb,amsthm}
\usepackage[x11names, rgb]{xcolor}
\usepackage{fullpage}
\usepackage{verbatim}
\usepackage{tikz}
\usepackage{bbold,xspace}
\usepackage[textwidth=23mm,colorinlistoftodos]{todonotes}
\usetikzlibrary{calc,arrows,shapes,backgrounds,patterns,fit,decorations,decorations.pathmorphing}

\newtheorem{theorem}{Theorem}
\newtheorem{lemma}{Lemma}

\newtheorem{corollary}{Corollary}

\newcommand{\prefix}{\textrm{prefix}}
\newcommand{\suffix}{\textrm{suffix}}
\newcommand{\overlap}{\textrm{overlap}}

\DeclareMathOperator{\operatorClassNP}{NP}
\newcommand{\classNP}{\ensuremath{\operatorClassNP}}
\DeclareMathOperator{\operatorClassCoNP}{coNP}
\newcommand{\classCoNP}{\ensuremath{\operatorClassCoNP}}
\DeclareMathOperator{\operatorClassFPT}{FPT}
\newcommand{\classFPT}{\ensuremath{\operatorClassFPT}}

\newcommand{\defproblemu}[3]{
  \vspace{1mm}
\noindent\fbox{
  \begin{minipage}{0.95\textwidth}
  #1 \\
  {\bf{Input:}} #2  \\
  {\bf{Question:}} #3
  \end{minipage}
  }
  \vspace{1mm}
}

\begin{document}

\title{Parameterized Complexity of Superstring Problems\thanks{The research leading to these results has received funding from the 
Government of the Russian Federation (grant 14.Z50.31.0030).}}

\author{
Ivan Bliznets\thanks{St.~Petersburg Department of Steklov Institute of Mathematics of the Russian Academy of Sciences} \addtocounter{footnote}{-1}
\and
Fedor V. Fomin \footnotemark {} \thanks{Department of Informatics, University of Bergen, Norway}  \addtocounter{footnote}{-2} 
\and
Petr A. Golovach \footnotemark {}  \footnotemark \addtocounter{footnote}{-2}
\and
Nikolay Karpov \footnotemark \addtocounter{footnote}{-1}
\and
Alexander S. Kulikov \footnotemark 
\and
Saket Saurabh \footnotemark {} \thanks{Institute of Mathematical Sciences, Chennai, India}
}

\maketitle
\sloppy
\begin{abstract}
In the  \textsc{Shortest Superstring} problem we are given a set of  strings $S=\{s_1, \ldots, s_n\}$  and integer $\ell$ and the question is to decide whether there is a superstring $s$ of length at most  $\ell$ containing all strings of $S$ as substrings. We obtain several parameterized algorithms and complexity results for this problem. 

In particular, we give an algorithm which in time $2^{O(k)} \operatorname{poly}(n)$ finds a superstring of length at most $\ell$ containing at least $k$ strings of $S$. We  complement this by the lower bound showing that such a  parameterization does not admit a polynomial kernel up to some complexity assumption. We also obtain several results about ``below guaranteed values" parameterization of the problem. We show that parameterization by compression admits a polynomial   kernel while parameterization ``below matching" is hard.  
\end{abstract}

\section{Introduction}
We consider the \textsc{Shortest Superstring} problem defined as follows:\\

\defproblemu{\textsc{Shortest Superstring}}{A set of $n$ strings $S=\{s_1, \ldots, s_n\}$ over an alphabet $\Sigma$ and
a non-negative integer~$\ell$.}
{Is there a string $s$ of length at most $\ell$ containing all strings from $S$ as substrings?}\\
This is a well-known NP-complete problem \cite{GareyJ79} with a range of practical  applications from 
DNA assembly \cite{EvansW11} till  data compression \cite{GallantMS80}. Due to this fact approximation algorithms
for it are widely studied. The currently best known approximation guarantee $2\frac{11}{23}$ is
due to Mucha~\cite{mucha}. At the same time the best known exact algorithms run in roughly 
$2^n$ steps
and are known for more than 50 years already. More precisely, using known algorithms for the \textsc{Traveling Salesman} problem, \textsc{Shortest Superstring} can be solved either in time 
$O^*(2^n)$
and the same space
by dynamic programming over subsets \cite{bellman,held} or in time $O^*(2^n)$ and only
polynomial space
by inclusion-exclusion~\cite{karp1982dynamic,kohn1977}
(here, $O^*(\cdot)$ hides factors that are polynomial in the input length, i.e., $\sum_{i=1}^n|s_i|$).
Such algorithms can only be used in practice to solve instances of very moderate size. Stronger upper bounds are known for a special case when input strings have bounded 
length~\cite{golovnev2013solving,golovnev2014solving}. There are
  heuristic methods for solving 
\textsc{Traveling Salesman}, and hence also \textsc{Shortest Superstring}, 
they are efficient in practice, however  have no efficient provable guarantee on the running time 
(see, e.g., \cite{concordeTSP}).

In this paper, we study the \textsc{Shortest Superstring} problem from the parameterized complexity point of view.
This field studies the complexity of computational problems with respect not only to input size,
but also to some additional parameters and tries to identify parameters of input instances that make
the problem tractable. 
Interestingly, prior to our work, except observations following from the known reductions to \textsc{Traveling Salesman}, not much about the parameterized complexity of  \textsc{Shortest Superstring} was known. We refer to the  survey of  Bulteau  et al. \cite{BulteauHKN14} for a nice overview of known results on parameterized algorithms and complexity of strings problems. Thus our work can be seen as the first non-trivial step towards the study of this interesting and important   problem from the perspective of parameterized complexity.

\paragraph{Our results}
In this paper we study two types of parameterization for \textsc{Shortest Superstring} and present two kind of results. 
The first set of results concerns ``natural" parameterization of the problem. 
We consider the following  generalization
of \textsc{Shortest Superstring}:

\medskip
\defproblemu{\textsc{Partial Superstring}}{A collection (multiset) of strings $S$ over an alphabet $\Sigma$,  and   non-negative integers $k,\ell$.}
{Is there a string $s$ of length at most $\ell$ such that $s$ is a superstring of a collection of at least $k$ strings $S'\subseteq S$?}\\
If  $k=|S|$, then this is  \textsc{Shortest Superstring}.
Notice that $S$ can contain copies of the same string and a string of $S$ can be a substring of another string of the collection. For \textsc{Shortest Superstring}, such cases could be easily avoided, but for \textsc{Partial Superstring} it is natural to assume that we have such possibilities.

Here we show that \textsc{Partial Superstring} is fixed parameter tractable (FPT) 
when parameterized by $k$ or~$\ell$. We complement this result by showing that
it is unlikely that the problem admits a polynomial kernel with respect to these parameters.

\medskip

The second set of results concerns ``below guaranteed value" parameterization. 
 Note that
an obvious (non-optimal) superstring  of $S=\{s_1, \ldots, s_n\}$ is a string of length 
 $\sum_{i=1}^n|s_i|$ formed by 
concatenating  all  strings from $S$. 
For a superstring $s$ of $S$ the value $\sum_{i=1}^{n}|s_i|-|s|$ is called by \emph{compression
of $s$ with respect to~$S$}. Then finding a shortest superstring is equivalent to
  finding
an order of $s_1, \ldots, s_n$ such that the consecutive strings have the largest possible total overlap.
We first show that it is FPT with respect to $r$ to check whether one can achieve a compression at least $r$ by construction a kernel of size $O(r^4)$. 
We complement this result by a hardness result about ``stronger" parameterization.
Let us partition $n$ input strings into $n/2$ pairs such that the sum of the $n/2$ resulting overlaps is maximized.
Such a partition can be found in polynomial time by constructing a maximum weight matching in an auxiliary graph.
Then this total overlap provides a lower bound on the maximum compression (or, equivalently, an upper bound
on the length of a shortest superstring). We show that already deciding whether at least one additional
symbol can be saved beyond the maximum weight matching value is already NP-complete. 

\section{Basic definitions and preliminaries}\label{sec:defs}
{\bf Strings.} Let $s$ be a string. By $|s|$ is denoted the \emph{length} of $s$. By $s[i]$, where $1\leq i\leq |s|$, is denoted the $i$-th symbol of $s$, and $s[i,j]=s[i]\ldots s[j]$ for $1\leq i\leq j\leq |s|$. We assume that $s[i,j]$ is the empty string if $i>j$.
We denote $\prefix_i(s)=s[1,i]$ and $\suffix_i(s)=s[|s|-i+1,|s|]$  the \emph{$i$-th prefix} and \emph{$i$-th suffix} of $s$ respectively for $i\in\{1,\ldots,|s|\}$;
$\prefix_0(s)=\suffix_0(s)$ is the empty string.
Let $s,s'$ be strings. We write $s\subseteq s'$ to denote that $s$ is a \emph{substring} of $s'$. If $s\subseteq s'$, then $s'$ is a \emph{superstring} of $s$.
We write $s\subset s'$ and $s\supset s'$ to denote proper sub and superstrings.
For a collection of strings $S$, a string $s$ is a superstring of $S$ if $s$ is a superstring of each string in $S$.
The \emph{compression measure} of a superstring $s$ of a collection of strings $S$ is $\sum_{x\in S}|x|-|s|$.
 If $s\subseteq s'$, then $\overlap(s,s')=\overlap(s',s)=s$; otherwise, if $s\not\subseteq s'$ and $s'\not\subseteq s$, then $\overlap(s,s')=\suffix_r(s)=\prefix_r(s')$, where $r=\max\{i\mid 0\leq i\leq\min\{|s|,|s'|\}, \suffix_i(s)=\prefix_i(s')\}$. We denote by $ss'$ the \emph{concatenation} of $s$ and $s'$. 
For strings $s,s'$, we define the \emph{concatenation with overlap} $s\circ s'$ as follows. If $s\subseteq s'$, then $s\circ s'=s'\circ s=s'$. If $s\not\subseteq s'$ and $s'\not\subseteq s$, then $s\circ s'=\prefix_p(s)\overlap(s,s')\suffix_q(s')$, where $p=|s|-|\overlap(s,s')|$ and $q=|s'|-|\overlap(s,s')|$.

We need the following folklore property  of superstrings.

\begin{lemma}\label{lem:overlap}
Let $s$ be a superstring of a collection  $S$ of  strings. 
Let $S'=\{s_1,\ldots,s_n\}$ be a set of inclusion maximal pairwise distinct strings of $S$ such that each string of $S$ is a substring of a string from $S'$.
Let also  $s_i=s[p_i,q_i]$ for $i\in\{1,\ldots,n\}$ and assume that $p_1<\cdots<p_n$.
Then $s'=s_1\circ\cdots\circ s_n$ is a superstring of $S$ of length at most $|s|$. 
\end{lemma}


\medskip
\noindent
{\bf Graphs.}
We consider finite directed and undirected graphs without loops or multiple
edges. The vertex set of a (directed) graph $G$ is denoted by $V(G)$, 
the edge set of an undirected graph and the arc set of a directed graph $G$ is denoted by $E(G)$.
To distinguish edges and arcs, the edge with two end-vertices $u,v$  is denoted by $\{u,v\}$, and we write $(u,v)$ for the corresponding arc.
For an arc $e=(u,v)$, $v$ is the \emph{head} of $e$ and $u$ is the tail. 
%
Let 
$G$ be a directed graph.
For a vertex $v\in V(G)$, we say that $u$ is an \emph{in-neighbor} of $v$ if $(u,v)\in E(G)$. The set of all in-neighbors of $v$ is denoted by $N_G^-(v)$. The \emph{in-degree} $d_G^-(v)=|N_G^-(v)|$.
Respectively, $u$ is an  \emph{out-neighbor} of $v$ if $(v,u)\in E(G)$, the set of all out-neighbors of $v$ is denoted by $N_G^+(v)$, and the \emph{out-degree} $d_G^+(v)=|N_G^+(v)|$.
For a directed graph $G$, a (directed) \emph{trail} of \emph{length} $k$ is a sequence $v_0,e_1,v_1,e_2,\ldots,e_k,v_k$ of vertices and arcs  of $G$ such that $v_0,\ldots,v_k\in V(G)$, $e_1,\ldots,e_k\in E(G)$,
the arcs  $e_1,\ldots,e_k$ are pairwise distinct, and 
for $i\in\{1,\ldots,k\}$, $e_i=(v_{i-1},v_i)$. 
We omit the word ``directed'' if it does not create a confusion. Slightly abusing notations we often write a trail as a sequence of its vertices $v_0,\ldots,v_k$ or arcs $e_1,\ldots,e_k$.
If $v_0,\ldots,v_k$ are pairwise distinct, then $v_0,\ldots,v_k$ is a (directed) path. Recall that a path of length $|V(G)|-1$ is a \emph{Hamiltonian} path.
For an undirected graph $G$, 
a set $U\subseteq V(G)$ is a \emph{vertex cover} of $G$ if for any edge $\{u,v\}$ of $G$, $u\in U$ or $v\in U$. A set of edges $M$ with pairwise distinct end-vertices is a \emph{matching}.

We consider the following auxiliary problem:\\
\defproblemu{\textsc{Long Trail}}{A directed graph $G$ and a non-negative integer $\ell$.}
{Is there a trail of length at least $\ell$ in $G$?}

\begin{lemma}\label{lem:trail}
\textsc{Long Trail} is \classNP-complete. In particular, the problem is \classNP-complete if $\ell=|V(G)|-~1$.
\end{lemma}

\begin{proof}
We reduce the \textsc{Hamiltonian Path} problem for directed graphs that is well known to be \classNP-complete (see, e.g., \cite{GareyJ79}). Let $G$ be a directed graph with $n$ vertices. We construct the graph $G'$ as follows.
\begin{itemize}
\item For each $v\in V(G)$, construct two vertices $v^-,v^+$ and an arc $(v^-,v^+)$.
\item For each $(u,v)\in E(G)$, construct an arc $(u^+,v^-)$.
\item Construct two vertices $s,t$ and for each $v\in V(G)$, construct arcs $(s,v^-),(v^+,t)$. 
\end{itemize}
We claim that $G'$ has a trail of length at least $2n+1=|V(G')|-1$ if and only if $G$ has a Hamiltonian path.

Suppose that $G$ has a Hamiltonian path $v_1,\ldots,v_n$. Then the trail $s,v_0^-,v_0^+v_1^-v_1^+,\ldots,v_n^-,v_n^+,t$ in $G'$ has length $2n+1$.

Assume that $G'$ has a trail $P$ of 
length at least $2n+1$. Without loss of generality we can assume that $s$ is the first vertex of $P$ and $t$ is the last. To see it, suppose that $x\neq s$ is the first vertex of $P$. Notice that 
$s$ is not in $P$, because $d_{G'}^-(s)=0$. If $x=v^-$ for $v\in V(G)$, then  we can consider the extended trail $s,(s,x),P$. If  $x=v^+$ for $v\in V(G)$, then let $u^-$ be the next vertex in $P$ after $x$. We consider the path $P'$ obtained from $P$ by the replacement of $x$ and $(x,u^-)$ by $s$ and $(s,u^-)$ respectively. Clearly, $P'$ has the same length as $P$. By the symmetric arguments, we obtain that we can assume that $t$ is the last vertex of $P$. We have that any vertex of $G'$ occurs exactly once in $P$, because
$d_{G'}^-(s)=d_{G'}^+(t)=0$ and $d_{G'}^+(v^-)=d_{G'}^-(v^+)=1$ for $v\in V(G)$.    
Moreover,  for each vertex $v\in V(G)$, $(v^-,v^+)$ in $P$, because $v^-$ is the unique in-neighbor of $v^+$ and $v^+$ is the unique out-neighbor of $v^-$ respectively for $v\in V(G)$. Hence, $P$ can be written as $s,v_0^-,v_0^+v_1^-v_1^+,\ldots,v_n^-,v_n^+,t$ for $v_1,\ldots,v_n\in V(G)$. It remains to observe that $v_1,\ldots,v_n$ is a Hamiltonian path in $G$.
\end{proof}

\smallskip
\noindent
{\bf Parameterized Complexity.}
Parameterized complexity is a two dimensional framework
for studying the computational complexity of a problem. One dimension is the input size
and another one is a parameter. We refer to the books of Downey and Fellows~\cite{DowneyF13},
Flum and Grohe~\cite{FlumG06}, and   Niedermeier~\cite{Niedermeierbook06} for  detailed introductions  to parameterized complexity. 

Formally, a parameterized problem $\mathcal{P}\subseteq\Sigma^*\times\mathbb{N}$, where $\Sigma$ is a finite alphabet, i.e., an instance of $\mathcal{P}$ is a pair $(I,k)$ for $I\in \Sigma^*$ and $k\in\mathbb{N}$, where $I$ is an input and $k$ is a parameter.
It is said that a problem is \emph{fixed parameter tractable} (or \classFPT), if it can be solved in time $f(k)\cdot |I|^{O(1)}$ for some function $f$. 
A \emph{kernelization} for a parameterized problem is a polynomial algorithm that maps each instance $(I,k)$  to an instance $(I',k')$ such that 
\begin{itemize}
\item[i)] $(I,k)$ is a yes-instance if and only if $(I',k')$ is a yes-instance of the problem, and
\item[ii)] the size of $I'$ and $k'$ are bounded by $f(k)$ for a computable function $f$. 
\end{itemize}
The output $(I',k')$ is called a \emph{kernel}. The function $f$ is said to be a \emph{size} of a kernel. Respectively, a kernel is \emph{polynomial} if $f$ is polynomial. 
While a parameterized problem is \classFPT{} if and only if it has a kernel, it is widely believed that not all \classFPT{} problems have polynomial kernels.

 In particular, Bodlaender, Jansen and Kratsch~\cite{BodlaenderJK14} introduced techniques that allow to show that a parameterized problem has no polynomial kernel unless  $\classNP\subseteq\classCoNP/\text{\rm poly}$.

Let $\Sigma$ be a finite alphabet. An equivalence relation $\mathcal{R}$ on the set of strings $\Sigma^*$ is called a \emph{polynomial equivalence relation} if the following two conditions hold:
\begin{itemize}
\item[i)] there is an algorithm that given two strings $x,y\in\Sigma^*$ decides whether $x$ and $y$ belong to
the same equivalence class in time polynomial in $|x|+|y|$,
\item[ii)] for any finite set $S\subseteq\Sigma^*$, the equivalence relation $\mathcal{R}$ partitions the elements of $S$ into a
number of classes that is polynomially bounded in the size of the largest element of $S$.
\end{itemize}

Let $L\subseteq\Sigma^*$ be a language, let $\mathcal{R}$ be a polynomial
equivalence relation on $\Sigma^*$, and let $\mathcal{P}\subseteq\Sigma^*\times\mathbb{N}$   
be a parameterized problem.  An \emph{OR-cross-composition of $L$ into $\mathcal{P}$} (with respect to $\mathcal{R}$) is an algorithm that, given $t$ instances $x_1,x_2,\ldots,x_t\in\Sigma^*$ 
of $L$ belonging to the same equivalence class of $\mathcal{R}$, takes time polynomial in
$\sum_{i=1}^t|x_i|$ and outputs an instance $(y,k)\in \Sigma^*\times \mathbb{N}$ such that:
\begin{itemize}
\item[i)] the parameter value $k$ is polynomially bounded in $\max\{|x_1|,\ldots,|x_t|\} + \log t$,
\item[ii)] the instance $(y,k)$ is a yes-instance for $\mathcal{P}$ if and only if at least one instance $x_i$ is a yes-instance for $L$ and $i\in\{1,\ldots,t\}$.
\end{itemize}
It is said that $L$ \emph{OR-cross-composes into} $\mathcal{P}$ if a cross-composition
algorithm exists for a suitable relation $\mathcal{R}$.

In particular, Bodlaender, Jansen and Kratsch~\cite{BodlaenderJK14} proved the following theorem.
\begin{theorem}[\cite{BodlaenderJK14}]\label{thm:BJK}
If an \classNP-hard language $L$ OR-cross-composes into the parameterized problem $\mathcal{P}$,
then $\mathcal{P}$ does not admit a polynomial kernelization unless
$\classNP\subseteq\classCoNP/\text{\rm poly}$.
\end{theorem}

We use randomized algorithms for our problems. Recall that a \emph{Monte Carlo algorithm} is a randomized algorithm whose running time is deterministic, but whose output may be incorrect with a certain (typically small) probability. A Monte-Carlo algorithm is \emph{true-biased} (\emph{false-biased} respectively) if it always returns a correct answer when it returns a yes-answer (a no-answer respectively).
 
\section{\classFPT-algorithms for Partial Superstring}
In this section we show that \textsc{Partial Superstring} is \classFPT, when parameterized by $k$ or $\ell$. For technical reasons, we consider the following variant of the problem with weights:\\
\defproblemu{\textsc{Partial Weighted Superstring}}{A collection of strings $S$ over an alphabet $\Sigma$ with a weight function $w\colon S\rightarrow \mathbb{N}_0$, and   non-negative integers $k,\ell$ and $W$.}
{Is there a string $s$ of length at most $\ell$ such that $s$ is a superstring of a collection of $k$ strings $S'\subseteq S$ with $w(S')\geq W$?}\\
Clearly, if $w\equiv 1$ and $W=k$, then we have the  \textsc{Partial Superstring} problem.

\begin{theorem}\label{thm:FPT-weight}
\textsc{Partial Weighted Superstring} can be solved  
in time $O((2e)^{k}\cdot kn^2m\log W)$ by  a true-biased Monte-Carlo algorithm
and in time $(2e)^{k}k^{O(\log k)}\cdot n^2\log n\cdot m\log W$ by a deterministic algorithm for a collection of $n$ strings of length at most $m$.
\end{theorem}

\begin{proof}
First, we describe the randomized algorithm and then explain how it can be derandomized. The algorithm uses the color coding technique proposed by Alon, Yuster and Zwick~\cite{AlonYZ95}. 

If $\ell\geq km$, then the problem is trivial, as the concatenation of any $k$ strings of $S$ has length at most $\ell$ and we can greedily choose $k$ strings of maximum weight. Assume that $\ell< km.$

We color the strings of $S$ by $k$ colors $1,\ldots,k$ uniformly at random independently from each other.
Now we are looking for a string $s$ that is a superstring of $k$ strings of maximum total weight that have pairwise distinct colors. 

To do it, we apply the dynamic programming across subsets. 
For simplicity, we explain only how to solve the existence problem, but our algorithm can be modified to find a colorful superstring as well.
For $X\subseteq \{1,\ldots,k\}$, a string $x\in S$ and a positive integer $h\in\{1,\ldots,\ell\}$,
the algorithm computes the maximum weight $W(X,x,h)$ of a string $s$ of length at most $h$ such that
\begin{itemize}
\item[i)] $s$ is a superstring of a collection of $k'=|X|$ strings $S'\subseteq S$ of pairwise distinct colors from $X$, 
\item[ii)] $x$ is inclusion maximal string of $S'$ and $x=\suffix_{|x|}(s)$. 
\end{itemize}
If such a string $s$ does not exist, then $W(X,x,h)=-\infty$.

We compute the table of values of $W(X,x,h)$ consecutively for $|X|=1,\ldots,k$. To simplify computations, we assume that $W(X,x,h)=-\infty$ for $h< 0$.
If $|X|=1$, then for each string $x\in S$, we set $W(X,x,h)=w(x)$ if $x$ is colored by the unique color of $X$ and $|x|\leq h$. In all other cases $W(X,x,h)=-\infty$.
Assume that $|X|=k'\geq 2$ and the values of $W(X',x,h)$ are already computed if $|X'|<k'$. 
Let
$$W'=\max\{W(X\setminus\{c\},x,h)+w(y)\mid y\subseteq x\text{ has  color }c\in X\},$$
and
$$W''=\max\{W(X\setminus \{c\},y,h-|x|+|\overlap(y,x)|)+w(x)\mid x\not\subseteq y,y\not\subseteq x\},$$
where $c$ is the color of $x$; we assume that $W'=-\infty$ if there is no substring $y$ of $x$ of color $c\in X$, and 
$W''=-\infty$  if every string $y$ is a sub or superstring of $x$. 
We set $W(X,x,h)=\max\{W',W''\}$.

We show that $\max\{W(\{1,\ldots,k\},x,\ell)\mid x\in S\}$ is the maximum weight  of $k$ strings of $S$ colored by distinct colors that have a superstring of length at most $\ell$; if this value equals $-\infty$, then there is no string of length at most $\ell$ that is a superstring of $k$ string of $S$ of distinct colors.

To prove this, it is sufficient to show that the values $W(X,x,h)$ computed by the algorithms are the maximum weights of strings of length at most $h$ that satisfy (i) and (ii). The proof is by induction on the size of $|X|$. It is straightforward to verify that it holds if $|X|=1$. Assume that $|X|>1$ and the  claim holds for sets of lesser size. 
Denote by  $W^*(X,x,h)$ the maximum weight of a string $s$ of length at most $h$ that satisfies (i) and (ii).  
By the description of the algorithm, $W^*(X,x,h)\geq W(X,x,h)$. We show that  $W^*(X,x,h)\leq W(X,x,h)$. 

Let $S'$ be a collection of $k'$ strings of pairwise distinct colors from $X$ that have $s$ as a superstring. Denote by $S''$ a set of inclusion maximal distinct strings of $S'$ that contains $x$ such that every string of $S'$ is a substring of a string of $S''$. Assume that $S''=\{x_1,\ldots,x_r\}$ and $x_i=s[p_i,q_i]$ for $i\in\{1,\ldots,r\}$.
Clearly, $x=x_r$. 

Suppose that there is $y\in S'\setminus\{x\}$ such that $y\subseteq x$. Let $c\in X$ be a color of $y$. Then $s$ is a superstring of $S'\setminus\{y\}$ and the total weight of these string is $W^*(X,x,h)-w(y)$. By induction, $W^*(X,x,h)-w(y)\leq W(X\setminus\{c\},x,h)$ and we have that $W^*(X,x,h)\leq W(X\setminus\{c\},x,h)+w(y)\leq W'\leq W(X,x,h)$.

Suppose now that $S'\setminus\{x\}$ does not contain substrings of $x$. Then $r\geq 2$. Let $y=s_{r-1}$ and $s'=s[1,q_{i-1}]$.  Observe that $y=\suffix_{|y|}(s')$.
Notice that $s'$ is a superstring of $S''\setminus x$. 
Because $S'\setminus \{x\}$ has no substrings of $x$, every string in $S'\setminus\{x\}$ is a substring of any superstring of $S''\setminus\{x\}$  and, therefore, $s'$ is a superstring of $S'\setminus\{x\}$ of length at most $|s|-|x|+|\overlap(y,x)|\leq h-|x|+|\overlap(y,x)|$.  The weight of $S'\setminus\{x\}$ is $W^*(X,x,h)-w(x)$. By induction, 
$W^*(X,x,h)-w(x)\leq W(X\setminus\{c\},y,h-|x|+|\overlap(y,x)|)$. Hence $W^*(X,x,h)\leq W(X\setminus\{c\},y,h-|x|+|\overlap(y,x)|)+w(x)\leq W''\leq W(X,x,h)$.

To evaluate the running time of the dynamic programming algorithm, observe that we can check whether $y$ is a substring of $x$ or find $\overlap(y,x)$ in time $O(m)$ using, e.g., the algorithm of Knuth, Morris, and Pratt~\cite{KnuthMP77}, 
and we can construct the table of the overlaps and their sizes in time $O(n^2m)$. 
Hence, for each $X$, the values $W(X,x,h)$ can be computed in time $O(n^2 k m\log W)$, as 
$h\leq \ell<km$. 
Therefore, the running time is $O(2^k\cdot n^2 k m\log W)$. 

We proved that an optimal colorful solution can be found in time  $O(2^k\cdot n^2 k m\log W)$. Using the standard color coding arguments (see~\cite{AlonYZ95}), we obtain that it is sufficient to consider $N=e^k$ random colorings of $S$ to claim that with probability $\alpha>0$, where $\alpha$ is a constant that does not depend on the input size and the parameter, we get a coloring for which $k$ string of $S$ that have a superstring of length at most $\ell$ and the total weight at least $W$ are colored by distinct colors if such a string exists.  It implies that 
\textsc{Partial Weighted Superstring} can be solved  
in time $O((2e)^{k}\cdot k n^2 m\log W)$ by  our randomized algorithm.

To derandomize the algorithm, we apply the  technique proposed by Alon, Yuster and Zwick~\cite{AlonYZ95}
using the $k$-perfect hash functions constructed by Naor, Schulman and Srinivasan~\cite{NaorSS95}. 
The random colorings are replaced by the family of at most $e^{k}k^{\log k}\log n$ hash functions $c\colon S\rightarrow\{1,\ldots,k\}$ that have the following property:
there is a hash function $c$ that colors 
$k$ string of $S$ that have a superstring of length at most $\ell$ and the total weight at least $W$ by distinct colors if such a string exists.
It implies that \textsc{Partial Weighted Superstring} can be solved in time $(2e)^k k^{O(\log k)}\cdot n^2\log n\cdot m\log W$ deterministically.
\end{proof}

Because \textsc{Partial Superstring} is a special case of \textsc{Partial Weighted Superstring}, Theorem~\ref{thm:FPT-weight} implies that this problem is \classFPT\ when parameterized by $k$. We show that the same holds if we parameterize the problem by $\ell$.

\begin{corollary}\label{cor:FPT-ell}
\textsc{Partial Superstring} is \classFPT\ when parameterized by $\ell$.
\end{corollary}

\begin{proof}
Consider an instance $(S,k,\ell)$ of \textsc{Partial Superstring}.  Recall that $S$ can contain several copies of the same string. We construct a set of weighted strings $S'$ by replacing a string $s$ that occurs $r$ times in $S$ by the single copy of $s$ of weight $w(s)=r$. Let $W=k$.
Observe that there is a string $s$ of length at most $\ell$ such that $s$ is a superstring of a collection of at least $k$ strings of $S$ if and only if there a string $s$ of length at most $\ell$ such that $s$ is a superstring of a set of  strings of $S'$ of total weight at least $W$.  A string of length at most $\ell$ has at most $\ell(\ell-1)/2$ distinct substrings. We consider the instances 
$(S',w,k',\ell,W)$ of \textsc{Partial Weighted Superstring} for $k'\in\{1,\ldots,\ell(\ell-1)/2\}$. For each of these instances, we solve the problem using Theorem~\ref{thm:FPT-weight}.
It remains to observe that  there is a string $s$ of length at most $\ell$ such that $s$ is a superstring of a set of  strings of $S'$ of total weight at least $W$ if and only if one of the instances 
$(S',w,k',\ell,W)$ is a yes-instance of \textsc{Partial Weighted Superstring}.
\end{proof}

We complement the above algorithmic results by showing that we hardly can expect that \textsc{Partial Superstring} has a polynomial kernel when parameterized by $k$ or $\ell$. 
\begin{theorem}\label{thm:no-kernel}
\textsc{Partial Superstring} does not admit a  polynomial kernel when parameterized by $k+m$ or $\ell+m$ for strings of length at most $m$ over the alphabet $\Sigma=\{0,1\}$ unless $\classNP\subseteq\classCoNP/\text{\rm poly}$.
\end{theorem}

 \begin{proof}[Theorem~\ref{thm:no-kernel}]
We show that \textsc{Long Trail} OR-cross-composes into \textsc{Partial Superstring}. Recall that \textsc{Long Trail} was shown to be \classNP-complete in Lemma~\ref{lem:trail}.

We assume that two instances $(G,\ell)$ and $(G',\ell')$ of  \textsc{Long Trail} are equivalent if $|V(G)|=|V(G')|$ and $\ell=\ell'$.
Consider equivalent instances $(G_i,\ell)$ of   \textsc{Long Trail} for $i\in\{1,\ldots,t\}$. Let $V(G_i)=\{v_1^i,\ldots,v_n^i\}$ for $i\in\{1,\ldots,t\}$. 
Let $r=\max\{\lfloor \log n\rfloor,\lfloor \log t\rfloor\}+2$. Denote by $x_i$ the string of length $r$ that encodes a positive integer $i$ in binary for $i\leq 2^{r}-1$.  
Let $x^*=x_i$ for $i=2^r-1$, i.e., $x^*=\text{'}1\ldots1\text{'}$. Notice that if $i\leq \max\{n,t\}$, then the first symbol of $x_i$ is '0'. 
For each arc $e=(v^i_p,v^i_q)$ of $G_i$, we construct a string $s_e=x_ix^*x_ix_px_ix^*x_ix_qx_ix^*x_i$. 
Clearly, $|s_e|=11r$.
We define
$$S=\{s_e\mid e\in E(G_i),1\leq i\leq t\}$$ and let $k=\ell$, $\ell'=4r\ell+7r$. 
We claim that there is $i\in\{1,\ldots,t\}$ such that $G_i$ has a trail of length $\ell$ if and only if there is a string $s$ of length at most $\ell'$ that is a superstring of $k$ strings of $S$.

Suppose that  there is $i\in\{1,\ldots,t\}$ such that $G_i$ has a  trail  $e_1,\ldots,e_{\ell}$.  Consider $s=s_{e_1}\circ\ldots\circ s_{e_\ell}$. Because the length of each $s_{e_i}$ is $11r$ and $|\overlap(s_{e_{i-1}},s_{e_i})|\geq 7r$, we obtain that $|s|\leq 11r\ell-7r(\ell-1)=\ell'$. Hence, $s$ is a string of length at most $\ell'$ that is a superstring of $k=\ell$ strings.

Assume now that there is a string $s$ of length at most $\ell'$ that is a superstring of $k$ strings of $S$. Because no string of $S$ is a substring of another one, we can assume that 
$s=s_{e_1}\circ\ldots\circ s_{e_k}$ for some $e_1,\ldots,e_k\in E(G_1)\cup\ldots\cup E(G_t)$ by Lemma~\ref{lem:overlap}.  
We use the following properties of the overlap of two strings $s_e,s_{e'}\in S$. Recall that if $e=(v^i_p,v^i_q)$ of $G_i$, then $s_e=x_ix^*x_ix_px_ix^*x_ix_qx_ix^*x_i$, $x^*=\text{'}1\ldots1\text{'}$ and 
the first symbol of $x_i$ is '$0$'. It implies that $|\overlap(s_e,s_{e'})|\leq 7r$ and  $|\overlap(s_e,s_{e'})|= 7r$ if and only if $e,e'\in E(G_i)$ for some $i\in\{1,\ldots,t\}$
and $e=(v^i_p,v_q^i)$, $e'=(v^i_q,v^i_z)$ for some $p,q,z\in\{1,\ldots,n\}$. Since $|s|\leq \ell'=4r\ell+7r$ and $k=\ell$, $|\overlap(s_{e_{j-1}},s_{e_j})|=7r$ for $j\in\{2,\ldots,k\}$.
Hence, $e_1,\ldots,e_\ell$  is a  trail in some $G_i$.

It remains to observe that $k+m=O(n+\log t)$ and $\ell'+m=O((n+\log t)^2)$ to complete the proof. 
\end{proof}

\section{Shortest Superstring below guaranteed values}\label{sec:below}
In this section we discuss \textsc{Shortest Superstring} parameterized by the difference between upper bounds for the length of a shortest superstring and the length of a solution superstring.  
For a collection of strings $S$, the length of the shortest superstring is trivially upper bounded by $\sum_{x\in S}|x|$. We show that \textsc{Shortest Superstring} admits a polynomial kernel when parameterized by the compression measure of a solution.

\begin{theorem}\label{thm:kernel}
\textsc{Shortest Superstring} admits a kernel of size $O(r^4)$ when parameterized by $r=\sum_{x\in S}|x|-\ell$.
\end{theorem}

\begin{proof}
Let $(S,\ell)$ be an instance of \textsc{Shortest Superstring}, $r=\sum_{x\in S}|x|-\ell$. First, we apply the following reduction rules for the instance.

\medskip
\noindent
{\bf Rule~1.} If there are distinct elements $x$ and $y$ of $S$ such that $x\subseteq y$, then delete $x$ and set $r=r-|x|$. If $r\leq 0$, then return a yes-answer and stop.

\medskip
\noindent
{\bf Rule~2.} If there is $x\in S$ such that for any $y\in S\setminus\{x\}$, $|\overlap(x,y)|=|\overlap(y,x)|=0$, then delete $x$ and set $\ell=\ell-|x|$. 
If $S=\emptyset$ and $\ell\geq 0$, then return a yes-answer and stop. If $\ell<0$, then return a no-answer and stop.

\medskip
\noindent
{\bf Rule~3.} If there are distinct elements  $x$ and $y$ of $S$ such that $|\overlap(x,y)|\geq r$, then return a yes-answer and stop. 

\medskip
It is straightforward to verify that these rules are \emph{safe}, i.e., by the application of a rule we either solve the problem or obtain an equivalent instance.
We exhaustively apply Rules~1--3. To simplify notations, we assume that $S$ is the obtained set of strings and $\ell$ and $r$ are the obtained values of the parameters.
Notice that all strings in $S$ are distinct and no string is a substring of another.
Our next aim is to bound the lengths of considered strings.

\medskip
\noindent
{\bf Rule~4.} If there is $x\in S$ with $|s|>2r$, then set $\ell=\ell-|x|+2r$ and $x=\prefix_r(x)\suffix_r(x)$. If $\ell<0$, then return a no-answer and stop.

\medskip
To see that the rule is safe, recall that $x$ is not a sub or superstring of any other string of $S$, and $|\overlap(x,y)|<r$ and $|\overlap(y,x)|<r$ for any $y\in S$ distinct from $x$ after the applications of Rule~3. As before, we apply Rule~4 exhaustively. 

Now we construct an auxiliary graph $G$ with the vertex set $S$ such that two distinct $x,y\in S$ are adjacent in $G$ if and only if $|\overlap(x,y)|>0$ or $|\overlap(y,x)|>0$. 
We greedily select a maximal matching $M$ in $G$ and apply the following rule.

\medskip
\noindent
{\bf Rule~5.} If $|M|\geq r$, then return a yes-answer and stop. 

\medskip
To show that the rule is safe, it is sufficient to observe that if $M=\{x_1,y_1\},\ldots,\{x_h,y_h\}$, $|\overlap(x_i,y_i)|>0$ for $i\in\{1,\ldots,h\}$ and $h\geq r$, then the string $s$ obtained by the consecutive concatenations with overlaps of   $x_1,y_1,\ldots,x_h,y_h$ and then all the other strings of $S$ in arbitrary order, then the compression measure of $s$ is at least $r$. 

Assume from now that we do not stop here, i.e., $|M|\leq r-1$. Let $X\subseteq S$ be the set of end-vertices of the edges of $M$ and $Y=S\setminus X$. Let $X=\{x_1,\ldots,x_h\}$. Clearly, $h\leq 2(r-1)$. Observe that $X$ is a vertex cover of $G$ and $Y$ is an independent set of $G$. 

For each ordered pair $(i,j)$ of distinct $i,j\in\{1,\ldots,h\}$, find an ordering $y_1,\ldots,y_t$ of the elements of $Y$ sorted by the decrease of $|\overlap(x_i,y_p)|+|\overlap(y_p,x_j)|$ for $p\in\{1,\ldots,t\}$. We construct the set $R_{(i,j)}$ that contains the first $\min\{2h,t\}$ elements of the sequence.

For each $i\in\{1,\ldots,h\}$, find an ordering $y_1,\ldots,y_t$ of the elements of $Y$ sorted by the decrease of $|\overlap(y_p,x_i)|$ for $p\in\{1,\ldots,t\}$. We construct the set $S_i$ that contains the first $\min\{2h,t\}$ elements of the sequence.

For each $i\in\{1,\ldots,h\}$, find an ordering $y_1,\ldots,y_t$ of the elements of $Y$ sorted by the decrease of $|\overlap(x_i,y_p)|$ for $p\in\{1,\ldots,t\}$. We construct the set $T_i$ that contains the first $\min\{2h,t\}$ elements of the sequence.

Let $$S'=X\cup \big(\bigcup_{(i,j),~i,j\in\{1,\ldots,h\},i\neq j}R_{(i,j)}\big)\cup\big(\bigcup_{i\in
\{1,\ldots,h\}}S_i \big)\cup\big(\bigcup_{i\in
\{1,\ldots,h\}}T_i \big).$$

\medskip
\noindent
{\bf Claim~$(*)$.} {\it There is a superstring $s$ of $S$ with the compression measure at least $r$ if and only if there is a superstring $s'$ of of $S'$ with the compression measure at least $r$. }

\begin{proof}[Proof of Claim~$(*)$]
If $s'$ is a superstring of $S'$ with the compression measure at least $r$, then the string $s$ obtained from $s'$ by the concatenation of $s'$ and the strings of $S\setminus S'$ (in any order) is a superstring of $S$ with the same compression measure as $s'$.

Suppose that $s$ is a shortest superstring of $S$ and the compression measure  at least $r$. By Lemma~\ref{lem:overlap}, $s=s_1\circ\ldots\circ s_n$, where $S=\{s_1,\ldots,s_n\}$.
Let $$Z=\{s_i\mid s_i\in Y, |\overlap(s_{i-1},s_i)|>0\text{ or }|\overlap(s_i,s_{i+1})|>0,1\leq i\leq n\};$$ 
we assume that $s_0,s_{n+1}$ are empty strings. 

We show that $|Z|\leq 2h$. Suppose that $s_i\in Z$. If $|\overlap(s_{i-1},s_i)|>0$, then $s_{i-1}\in X$, because $s_i\in Y$ and any two strings of $Y$ have the empty overlap.
By the same arguments, if   $|\overlap(s_i,s_{i+1})|>0$, then $s_{i+1}\in X$. Because $|X|=h$, we have that $|Z|\leq 2h$.

Suppose that the shortest superstring $s$ is chosen in such a way that $|Z\setminus S'|$ is minimum. We prove that $Z\subseteq S'$ in this case.
To obtain a contradiction, assume that there is $s_i\in Z\setminus S'$. We consider three cases.

\medskip
\noindent
{\bf Case 1.} $|\overlap(s_{i-1},s_i)|>0$ and $|\overlap(s_i,s_{i+1})|>0$. Recall that $s_{i-1},s_{i+1}\in X$ in this case. Since $s_i\notin S'$, $s_i\notin R_{(p,q)}$ for $x_p=s_{i-1}$ and $x_q=s_{i+1}$.
In particular, it means that $|R_{(p,q)}|=2h$. As $|Z|\leq 2h$ and $|R_{(p,q)}|=2h$, there is $s_j\in R_{(p,q)}$ such that $s_j\notin Z$, i.e., $|\overlap(s_{j-1},s_j)|=|\overlap(s_j,s_{j+1})|=0$. By the definition of $R_{(p,q)}$, $|\overlap(s_{i-1},s_j)|+|\overlap(s_j,s_{i+1})| \geq
 |\overlap(s_{i-1},s_i)|+|\overlap(s_i,s_{i+1})|$.  Consider 
$s^*=s_1\circ\ldots\circ s_{i-1}\circ s_j\circ s_i\ldots \circ s_{j-1}\circ s_i\circ s_j\circ \ldots \circ s_n$ assuming that $i<j$ (the other case is similar). Because 
$|\overlap(s_{i-1},s_j)|+|\overlap(s_j,s_{i+1})|\geq |\overlap(s_{i-1},s_i)|+|\overlap(s_i,s_{i+1})|$, $|s^*|\leq |s|$. Moreover, since $s$ is a shortest superstring of $S$, 
$|s|\geq |s^*|$ and, therefore, $|\overlap(s_{j-1},s_i)|=|\overlap(s_i,s_{j+1})|=0$. But then for  the set $Z^*$ constructed for $s^*$ in the same way as the set $Z$ for $s$, we obtain that 
$|Z^*\setminus S'|<|Z\setminus S'|$; a contradiction.

\medskip
\noindent
{\bf Case 2.} $|\overlap(s_{i-1},s_i)|=0$ and $|\overlap(s_i,s_{i+1})|>0$. Then $s_{i+1}\in X$.
 Since $s_i\notin S'$, $s_i\notin S_{p}$ for $x_p=s_{i+1}$ and $|S_p|=2h$. 
As $|Z|\leq 2h$ and $|S_p|=2h$, there is $s_j\in S_p$ such that $s_j\notin Z$, i.e., 
$|\overlap(s_{j-1},s_j)|=|\overlap(s_j,s_{j+1})|=0$. By the definition of $S_p$,
$|\overlap(s_j,s_{i+1})|\geq |\overlap(s_i,s_{i+1})|$.  
As in Case~1, consider $s^*$ obtained by the exchange of $s_i$ and $s_j$ in the sequence of strings that is used for the concatenations with overlaps.   
In the same way, we obtain a contradiction with the choice of $Z$, because for  the set $Z^*$ constructed for $s^*$ in the same way as the set $Z$ for $s$, we obtain that 
$|Z^*\setminus S'|<|Z\setminus S'|$.

\medskip
\noindent
{\bf Case 3.} $|\overlap(s_{i-1},s_i)|>0$ and $|\overlap(s_i,s_{i+1})|=0$. To obtain contradiction in this case, we use the same arguments as in Case~2 using symmetry. Notice that we should consider $T_p$ instead of $S_p$.

\medskip
Now  let $s'=s_{i_1}\circ\ldots\circ s_{i_p}$, where $s_{i_1},\ldots,s_{i_p}$ is the sequence of string of $S'$ obtained from $s_1,\ldots,s_n$ by the deletion of the strings of $S\setminus S'$. Because we have that $Z\subseteq S'$, the overlap of each deleted string with its neighbors is empty and, therefore, $s'$ has the same compression measure as $s$.
\end{proof}

To finish the construction of the kernel, we define $\ell'=\ell-\sum_{x\in S\setminus S'}|x|$ and apply the following rule that is safe by Claim $(*)$.

\medskip
\noindent
{\bf Rule~6.} If $\ell'<0$, then return a no-answer and stop. Otherwise, return the instance $(S',\ell')$ and stop. 

\medskip
Since $|X| =h \leq 2(r-1)$, $|S'|\leq h + h^2 \cdot 2h+h \cdot 2h+h\cdot 2h =2 h^3+4h^2+h =O(h^3) =O(r^3)$. Because each string of $S'$ has length at most $2r$, the kernel has size $O(r^4)$. 

It is easy to see that  Rules~1-3 can be applied in polynomial time.  Then graph $G$  and $M$ can be constructed in polynomial time and, trivially, Rule~5 demands $O(1)$ time. The sets $X$, $Y$, $R_{(i,j)}$, $S_i$ and $T_i$ can be constructed in polynomial time. Hence, $S'$ and $\ell'$ can be constructed in polynomial time. Because Rule~6 can be applied in time $O(1)$, we conclude that the kernel is constructed in polynomial time.
\end{proof}

Now we consider another upper bound for the length of the shortest superstring. Let $S$ be a collection of strings. We construct 
 an auxiliary weighted graph $G(S)$ with the vertex set $S$ by assigning the weight $w(\{x,y\})=\max\{|\overlap(x,y)|,|\overlap(y,x)|\}$
for any  two distinct $x,y\in S$. Let $\mu(S)$ be the size of a maximum weighted matching in $G$. Clearly, $G$ can be constructed in polynomial time and the computation of $\mu(G)$ is well known to be polynomial~\cite{Edmonds65}.
If $M=\{x_1,y_1\},\ldots,\{x_h,y_h\}$ and $|\overlap(x_i,y_i)|=w(\{x_i,y_i\})$ for $i\in\{1,\ldots,h\}$, then the string $s$ obtained by the consecutive concatenations with overlaps of   $x_1,y_1,\ldots,x_h,y_h$ and then (possibly) the remaining string of $S$ has the compression measure  at least $\mu(G)$. Hence, $\sum_{x\in S}|x|-\mu(G)$ is the upper bound for the length of the shortest superstring of $G$. We show that it is \classNP-hard to find a superstring that is shorter than this bound.

\begin{theorem}\label{thm:ss_NP_mat}
\textsc{Shortest Superstring} is \classNP-complete for $\ell=\sum_{x\in S}|x|-\mu(S)-1$ even if restricted to the alphabet $\Sigma=\{0,1\}$.
\end{theorem} 

\begin{proof}
We reduce \textsc{Long Trail} that was shown to be \classNP-complete in Lemma~\ref{lem:trail} for $\ell=|V(G)|-1$. Let $(G,\ell)$ be an instance of the problem, $n=|V(G)|=\ell+1$. We assume that $n\geq 2^6=64$. Let $V(G)=\{v_1,\ldots,v_n\}$ and $E(G)=\{e_1,\ldots,e_m\}$.
Let also $p=\lceil (n-1)/3\rceil$ and $q=n-1-2p$. Denote by $z=\text{'}01\ldots1\text{'}$ and $z^*=\text{'}1\ldots1\text{'}$ the strings of length $p$ such that the first symbol of $z$ is '0' and all the other symbols are '1'-s and $z^*$ is a strings of '1'-s.
For a positive integer $i\leq 2^{q-1}-1$, denote by $x_i$ the string of length $q-1$ that encodes $i$ in binary and by $y_i$ the string of length $q$ that encodes $2i$. 
Notice that $q\geq n/3-4$ and $\log n^2\leq q-3$, because $n\geq 2^6$. Hence, the first symbols of $x_i$ and $y_i$ are '0' if $i\leq n^2$. Observe also that the last symbol of each $y_i$ is '0'. 
For each $h\in\{1,\ldots,m\}$, we consider the arc $e_h=(v_i,v_j)$ of $G$ and construct two strings:
\begin{itemize}
\item $s_h=zy_hz^*zx_iz^*zx_jz^*$, 
\item $s_h'= zx_iz^*zx_jz^*zy_hz^*$.
\end{itemize}
We define $S=\{s_h,s_h'\mid 1\leq h\leq m\}$.

We need the following properties of the strings of $S$.
\begin{itemize}
 \item[i)] For $h\in\{1,\ldots,m\}$, $|\overlap(s_h,s_h')|=2(n-2)$ and $|\overlap(s_h',s_h)|=n-1$.
 \item[ii)] For distinct $h,h'\in\{1,\ldots,m\}$, $|\overlap(s_h,s'_{h'})|=n-2$ if the head of $e_h$ coincides with the tail of $e_{h'}$ and  $|\overlap(s_h,s'_{h'})|=0$ otherwise.
 \item[iii)] For distinct $h,h'\in\{1,\ldots,m\}$, $|\overlap(s'_h,s_{h'})|=|\overlap(s_h,s_{h'})|=|\overlap(s'_h,s'_{h'})|=0$.
\end{itemize}
These properties immediately follow from the definition of $s_h,s_h'$ and the facts that $|z|=|z^*|\geq |y_h| =|x_i|+1=|x_j|+1$, the strings $z,y_h,x_i,x_j$ start with '0', the last symbol of $y_h$ is \lq{}0',  and $z=\text{'}01\ldots1\text{'}$, $z^*=\text{'}1\ldots1\text{'}$. It is sufficient to notice that if the overlap of two strings is not empty, then the $p$-th prefix and suffix of the overlap is always $z$ and $z^*$ respectively.

Now we consider the weighted graph $G(S)$ and observe that $M=\{\{s_h,s_h'\}\mid 1\leq h\leq m\}$ is a maximum weight matching in $G(S)$ and $\mu(S)=2(n-2)m$ by (i)--(iii).

We claim that $G$ has a  trail of length at least $\ell=n-1$ if and only if $S$ has a superstring of length at most $\ell'=\sum_{x\in S}|x|-\mu(S)-1$.

Suppose that the sequence of arcs $e_{i_1},\ldots,e_{i_{\ell}}$ composes a trail in $G$. Let $\{e_{j_1},\ldots,e_{j_{m-\ell}}\}=S\setminus\{e_{i_1},\ldots,e_{i_{\ell}}\}$.
Consider $$s=s_{i_1}'\circ s_{i_1}\circ \ldots\circ s_{i_\ell}'\circ s_{i_\ell}\circ s_{j_1}\circ s_{j_1}'\circ\ldots\circ   s_{j_{m-\ell}}\circ s_{j_{m-\ell}}' .$$
Since $|\overlap(s_{i_h}',s_{i_h})|=n-1$ for $h\in\{1,\ldots,\ell\}$ by (i), $|\overlap(s_{i_{h-1}},s_{i_h}')|=n-2$ for $h\in\{2,\ldots,\ell\}$ by (ii) and
$|\overlap(s_{j_h}, s_{j_h}')|=2(n-2)$ by (i), the compression measure of $s$ is $t=(n-1)\ell+(n-2)(\ell-1)+ 2(n-2)(m-\ell)$ and 
$t-\mu(S)=(n-1)\ell-(n-2)(\ell+1)=1$. Hence, $s$ is a superstring of $S$ of length at most $\ell'$.

Assume that $s$ is a shortest superstring of $S$ and $|s|\leq \ell'$. By Lemma~\ref{lem:overlap}, we can assume that $s$ is obtained from a sequence $\sigma$ of the strings of $S$ by the concatenations with overlaps.

We show that for every $h\in\{1,\ldots,m\}$, either $s_h,s_h'$ or $s_h',s_h$ are consecutive in $\sigma$. To obtain a contradiction, assume first that for some $h\in\{1,\ldots,m\}$, $s_h$ occurs in $\sigma$ before $s_h'$ but these strings are not consecutive.  Let $a$ be the predecessor of $s_h$, $b$ be a predecessor of $s_h'$ and $c$ be a successor of $s_h'$ in $\sigma$; if $s_h$ is the first element of $\sigma$ or $s_h'$ is the last element, we assume that $a$ or $c$ is the empty string respectively. 
Then $|\overlap(a,s_h)|=|\overlap(s_h'),c|=0$ by (iii) and $|\overlap(b,s_h')|\leq n-2$ by (ii) and (iii). Consider the sequence 
$\sigma'$ obtained from $\sigma$ by the placement of $s_h'$ between $a$ and $s_h$. Because $|\overlap(s_h',s_h)|=n-1$ by (1), the string $s'$ obtained from $\sigma'$ by the concatenations with overlaps has length at  most $|s|-1$; a contradiction.  Suppose now that      
for some $h\in\{1,\ldots,m\}$, $s_h'$ occurs in $\sigma$ before $s_h$ but these strings are not consecutive. Let  $a$ be the successor of $s_h'$, $b$ be a predecessor of $s_h$ and $c$ be a successor of $s_h$ in $\sigma$; if $s_h$ is the  last element of $\sigma$, we assume that $c$ is the empty string. 
We have that  $|\overlap(s_h',a)|=|\overlap(b,s_h)|=0$ by (iii) and $|\overlap(s_h,c)|\leq n-2$ by (ii) and (iii). Consider the sequence 
$\sigma'$ obtained from $\sigma$ by the placement of $s_h$ between $s_h'$ and $a$. Because $|\overlap(s_h',s_h)|=n-1$ by (i), the string $s'$ obtained from $\sigma'$ by the concatenations with overlaps has length at  most $|s|-1$; a contradiction. 

We decompose $\sigma$ into inclusion maximal subsequences $\sigma_1,\ldots,\sigma_r$ such that the overlap between any two consecutive strings in each subsequence is not empty.
Because  either $s_h,s_h'$ or $s_h',s_h$ are consecutive in $\sigma$ for $h\in\{1,\ldots,m\}$ and  $|\overlap(s_h,s_h')|=2(n-2)$ and $|\overlap(s_h',s_h)|=n-1$ by (i), each pair $s_h,s_h'$ is in the same subsequence. In particular, it means that the number of elements in each subsequence is even. Let $n_i$ be the size of $\sigma_i$ and let $w_i$ be the string obtained by the concatenation with overlaps from $\sigma_i$ for $i\in\{1,\ldots,r\}$. Because $n_1+\ldots+n_r=2m$, $|M|=m$ and the compression measure of $s$ is at least $\mu(S)+1$, 
there is $i\in\{1,\ldots,r\}$ such that the compression measure $\alpha$ of $w_i$ is at least $n_i/2\cdot\mu(S)/m+1=n_i(n-2)+1$.

Suppose that $s_h,s_h'$ are in $\sigma_i$ for some $h\in\{1,\ldots,m\}$. Then they are consecutive. If $s_h$ has a predecessor $a$ in $\sigma$, then $|\overlap(a,s_h)|=0$, and if 
$s_h'$ has a successor $b$ in $\sigma$, then $|\overlap(s_h',b)|=0$ by (iii). Hence $\sigma_i=s_h,s_h'$ and $n_i=2$ in this case, but then by (i), $\alpha=2(n-2)<n_i(n-1)/2+1$; a contradiction. It follows that $w_i=s_{i_1}'\circ s_{i_1}\circ\ldots\circ  s_{i_k}'\circ s_{i_k}$, where distinct $i_1,\ldots,i_k\in\{1,\ldots,m\}$ and $k=n_i/2$. 
Since for $j\in\{2,\ldots,k\}$, the overlap between $s_{i_{j-1}}$ and $s_{i_j}'$ is not empty, $|\overlap(s_{i_{j-1}},s_{i_j}')|=n-2$ and the head of the arc $e_{i_{j-1}}$ is the tail of $e_{i_j}$.
Hence, $e_{i_1},\ldots,e_{i_k}$ is a  trail in $G$. By (i) and (ii), we have that $\alpha=k(n-1)+(k-1)(n-2)\geq 2k(n-2)+1$. Therefore, $k\geq n-1$, i.e., $G$ has a  trail of length at least $\ell=n-1$. 
\end{proof}

\bibliographystyle{splncs03}
\bibliography{superstring}

\end{document}